\documentclass[a4paper]{llncs}
\usepackage[utf8]{inputenc}

\usepackage{amsmath, amssymb}
\usepackage{hyperref}
\usepackage{xspace}
\usepackage{xcolor}

\newcommand{\N}{\ensuremath{\mathbb{N}}\xspace}
\newcommand{\Z}{\ensuremath{\mathbb{Z}}\xspace}

\newcommand\T{\mathcal{T}}
\newcommand\bimplies{\ensuremath{\rightarrow}}
\newcommand\TFrag{\mathcal{F}}

\newcommand\Array{\smt{Array}}
\newcommand\Int{\smt{Int}}
\newcommand\structure[1]{\textbf{#1}}
\newcommand\Arr{\ensuremath{\mathcal{A}}\xspace}
\newcommand\tarr{\ensuremath{\T_\Arr}\xspace}
\newcommand\store[3]{\ensuremath{#1\langle #2\lhd#3\rangle}\xspace}
\newcommand\select[2]{\ensuremath{#1[#2]}\xspace}

\newcommand\diff{\smt{diff}}
\newcommand\tint{\ensuremath{\T_\Z}\xspace}

\newcommand\mI{\ensuremath{I}}
\newcommand\I{$\mI$\xspace}
\newcommand\A{$A$\xspace}
\newcommand\B{$B$\xspace}
\newcommand\smt[1]{\mathsf{#1}}

\newcommand\Model{\ensuremath{\mathcal{M}}\xspace}

\title{Interpolation and the Array Property Fragment}
\author{Jochen Hoenicke \and Tanja Schindler}
\institute{University of Freiburg}

\begin{document}

\maketitle

\begin{abstract}
  Interpolation based software model checkers have been successfully
  employed to automatically prove programs correct.
  Their power comes from interpolating SMT solvers that check the
  feasibility of potential counterexamples and compute candidate
  invariants, otherwise.
  This approach works well for quantifier-free theories, like equality theory or linear arithmetic.

  For quantified formulas, there are SMT solvers that can decide expressive fragments of quantified formulas, e.\,g.,
  EPR, the array property fragment, and the finite almost uninterpreted fragment.
  However, these solvers do not support interpolation.
  It is already known that in general EPR does not allow for interpolation.
  In this paper, we show the same result for the array property fragment.
\end{abstract}

\section{Introduction}
\label{sec:introduction}
Several software model checkers~\cite{DBLP:conf/fsttcs/CassezMB14,DBLP:conf/tacas/CassezSRPSM17,DBLP:conf/tacas/DanglLW15,DBLP:conf/tacas/HeizmannCDGHLNM18,DBLP:conf/popl/HenzingerJMM04,DBLP:conf/fmcad/JovanovicD16,DBLP:conf/cav/KupriyanovF14,DBLP:conf/cav/McMillan06,DBLP:conf/tacas/NutzDMP15}
use interpolating SMT solvers for various subtasks of software verification.
In counterexample-based approaches, for instance, paths through the program that lead to an error are encoded as formulas.
If the formula is satisfiable, the resulting model can be translated into a concrete counterexample to correctness.
If the formula is unsatisfiable, a \emph{Craig interpolant} can be generated that serves to compute a candidate invariant.
These candidate invariants can then be checked for inductiveness.

The software model checkers are powered by interpolating SMT solvers.
Ideally, the fragment supported by the SMT solver is decidable and supports interpolation, to ensure completeness of the interpolation procedure.
Additionally, the fragment should be closed under the usual logical operations like conjunction and negation, to facilitate inductiveness checks in the solver.
Many quantifier-free fragments of SMT theories and their combinations are decidable and closed under logical operations and interpolation~\cite{DBLP:journals/tcs/McMillan05,DBLP:conf/cade/YorshM05,DBLP:journals/corr/abs-1204-2386}.

The full first-order logic is closed under all operations and interpolation, but it is not decidable.
There are decidable fragments, for example, EPR, the array property fragment, and the finite almost uninterpreted fragment.
For each of the fragments, the full fragment that supports $\exists\forall$-formulas is not closed under negation\footnote{For example $\exists i.\ \forall j.\ a[j] \leq a[i]$ is in the array property and the finite almost uninterpreted fragment (after Skolemisation), but its negation is not.}.
If restricted to the alternation-free fragment, EPR and the array property fragment are closed under all logical operations.
In~\cite{DBLP:conf/cav/DrewsA16}, Drews and Albarghouti show that the alternation free EPR fragment is not closed under interpolation.
We show in this paper a similar result for the (alternation free) array property fragment:
we give an example for an interpolation problem where the input formulas are in the array property fragment and prove that there
exists no interpolant within the array property fragment.

\section{Notation and Basic Definitions}
\label{sec:notation}
We use sorted first-order logic and the model-theoretic approach to define theories from the SMT-LIB standard~\cite{BarFT-RR-17}.
The \emph{sort symbols} $s$ with arity $\mathit{ar}(s)\geq 0$ inductively define the set of \emph{sorts}:
if $\sigma_1,\dots,\sigma_n$ are sorts and $\mathit{ar}(s)=n$, then $\sigma:= s\sigma_1\dots\sigma_n$ is a sort.
The \emph{function symbols} $f$ with rank $\sigma_1\dots\sigma_n\sigma$ define the set of \emph{terms}: 
if $t_1,\dots,t_n$ are terms of sort $\sigma_1,\dots,\sigma_n$ respectively, then $f(t_1,\dots,t_n)$ is a term of sort $\sigma$.
A \emph{signature} $\Sigma$ is given by a set of sort symbols and a set of ranked function symbols.
A \emph{$\Sigma$-formula} is a first-order formula built from the sorts and function symbols in $\Sigma$.
A \emph{$\Sigma$-structure} $\structure{A}$ maps sorts $\sigma$ to a non-empty set $\sigma^{\structure{A}}$ and function symbols $f$ with rank $\sigma_1\dots\sigma_n\sigma$ to a corresponding function $\sigma_1^{\structure{A}}\times \sigma_n^{\structure{A}} \rightarrow \sigma^{\structure{A}}$.
A \emph{theory} $\T$ is given by its signature $\Sigma$ and a class of $\Sigma$-structures, which are also called the \emph{models} of $\T$.
A \emph{theory fragment} $\TFrag$ of a theory $\T$ with signature $\Sigma$ is a subset of $\Sigma$-formulas.

The \emph{theory of arrays} \tarr is parameterized by a signature $\Sigma$ defining other sort symbols that can be used for index and element sorts.  
The signature of $\tarr[\Sigma]$ contains in addition a sort symbol $\Array$ of arity $2$.
For each index sort $\sigma_I$ and element sort $\sigma_E$, the sort $\sigma_A := \Array\ \sigma_I \sigma_E$ represents the sort of arrays with the given index and element sort.
The signature contains a select function \(\select{\cdot}{\cdot}\) of rank $\sigma_A \sigma_I \sigma_E$ and a store function \(\store{\cdot}{\cdot}{\cdot}\) of rank $\sigma_A \sigma_I \sigma_E\sigma_A$.
For array~\(a\), index~\(i\), and element~\(v\), \(\select{a}{i}\) returns the element stored in \(a\) at index \(i\), and \(\store{a}{i}{v}\) returns a fresh array that is a copy of \(a\) where the element at \(i\) is replaced by \(v\).
The models $\Model$ of $\tarr$ fix the meaning of $\Array$, $\select{\cdot}{\cdot}$, and $\store{\cdot}{\cdot}{\cdot}$ as follows.
\begin{align*}
  (\Array\ \sigma_I \sigma_E)^{\Model} &:= \sigma_I^{\Model} \rightarrow \sigma_E^{\Model}\\
  (\select{\cdot}{\cdot})^{\Model}(a,i)\hphantom{,v} &:= a(i)\\
  (\store{\cdot}{\cdot}{\cdot})^{\Model}(a, i, v) &:
     \begin{array}[t]{r@{}l}
         \sigma_I^{\Model} &{}\to \sigma_E^{\Model}\\ 
         j &{}\mapsto \begin{cases} a(j) & i\neq j\\
                                 v    & i = j
                   \end{cases}
     \end{array}
\end{align*}

The \emph{theory of integers} $\tint$ contains a sort symbol $\Int$ of arity $0$ as well as the usual arithmetic functions $+,-,*,<,\leq$ in its signature.
Its models $\Model$ define $\Int^\Model := \Z$ and fix the meaning of the arithmetic functions as usual.
In the following we use the combined theory $\T := \tarr + \tint$, where the signature contains all symbols from $\tarr$ and $\tint$ and the meaning of the theory symbols is defined as above.

An \emph{interpolation problem} \((A,B)\) is a pair of formulas where the conjunction \(A\land B\) is unsatisfiable.
Given an interpolation problem \((A,B)\), the symbols shared between \A and \B are called \emph{shared}, symbols only occurring in \A are called \emph{\A-local} and symbols only occurring in \B, \emph{\B-local}.
We call a term or formula \emph{shared} if it contains only shared symbols.
A \emph{Craig interpolant} for an interpolation problem \((A,B)\) is a formula \I such that
(i) \A implies \I in the theory $\T$,
(ii) \I and \B are $\T$-unsatisfiable and
(iii) \I is shared between \A and \B.

A theory fragment $\TFrag$ is \emph{closed under interpolation} if for each interpolation problem \((A,B)\) in $\TFrag$ there exists an interpolant \I in $\TFrag$.

\section{The Array Property Fragment}

The theory of arrays is often used in software verification to model heap memory or arrays in C, for instance.
The quantifier-free fragment of the theory of arrays is decidable, and there exist interpolation methods for an extension of the quantifier-free fragment of the theory of arrays~\cite{DBLP:journals/corr/abs-1204-2386,DBLP:conf/cade/HoenickeS18}.
However, in many verification tasks, it is not sufficient to consider quantifier-free formulas.
For instance, to prove correctness of a sorting algorithm, it is necessary to reason about all elements of an array within a given range.

Quantifiers are challenging as general First-Order logic with theories is undecidable.
The same holds for the theory of arrays with quantifiers.
The decidable array property fragment was introduced by Bradley~et~al. in~\cite{DBLP:conf/vmcai/BradleyMS06} as a subset of the \(\exists^*\forall^*_\Z\)-fragment of the theory $\T$ of arrays and integers.

An \emph{array property} is a quantified formula of the form
\[\forall \bar{j}.\, \varphi_I(\bar{j}) \bimplies \varphi_V(\bar{j})\]
where \(\bar{j}\) are index variables, and the form of the \emph{index guard} \(\varphi_I(\bar{j})\) and the \emph{value constraint} \(\varphi_V(\bar{j})\) are restricted syntactically as follows:
the index guard \(\varphi_I(\bar{j})\) consists of ground literals and literals containing quantified variables of the form
\[j \leq t, \quad t \leq j, \quad j \leq j', \quad j = t, \quad j = j'\]
where \(t\) is a ground term and \(j,j' \in \bar{j}\) are quantified variables.
The literals can be connected by \(\land\) and \(\lor\) but the literals containing quantified variables must not appear negated.
The value constraint \(\varphi_V(\bar{j})\) consists of ground literals and literals containing quantified index variables \(j\) only within array reads \(\select{a}{j}\).
Array reads on quantified variables must not be nested, i.e., \(\select{a}{j}\) must not occur in arguments of the select function \(\select{\cdot}{\cdot}\) or the store function \(\store{\cdot}{\cdot}{\cdot}\).

The \emph{array property fragment} of $\T$ consists of all Boolean combinations of array properties and quantifier-free formulae.
It is sufficiently expressive to describe properties such as sortedness or equality of arrays in a given range.
In the original presentation, one quantifier alternation is allowed, i.e., an array property can be existentially quantified.
However, the resulting fragment is not closed under negation which is crucial for interpolation-based invariant generation.

We follow here the more restricted definition in~\cite{DBLP:books/daglib/0019162} that does not allow alternation of quantifiers.
Under this restriction, the fragment is closed under negation.
However, it is not closed under interpolation\footnote{If one allows \(\exists^*\forall^*_\Z\)-formulae, one might find an interpolant in this form, but the negation does not lie in the fragment and hence cannot be used for inductiveness checks.}, i.e., there exist formulas within the array property fragment for which no interpolant in the array property fragment exists, as we show in the next section.

\section{Interpolation in the Array Property Fragment}
In the following, we show that the array property fragment as defined above is not closed under interpolation by giving a concrete counterexample.

\begin{example}\label{thm:ap_notclosed_ex}
  Consider the following interpolation problem in the array property fragment.
  \begin{align*}
	&A: \forall \smt{i}.\ \select{\smt{a}}{\smt{i}} < \select{\smt{b}}{\smt{k}}\\
	&B: \forall \smt{j}.\ \lnot (\select{\smt{a}}{\smt{l}} < \select{\smt{b}}{\smt{j}})
  \end{align*}
  Clearly, \(A \land B\) is unsatisfiable:
  \A implies that \(\select{\smt{a}}{\smt{l}} < \select{\smt{b}}{\smt{k}}\) must hold by instantiating \(\smt{i}\) with \(\smt{l}\), in contradiction to \(\lnot (\select{\smt{a}}{\smt{l}} < \select{\smt{b}}{\smt{k}})\) which results from instantiating \(\smt{j}\) with \(\smt{k}\) in \B.
  Possible interpolants are
  \[\mI_1 = \exists \smt{j}. \forall \smt{i}.\ \select{\smt{a}}{\smt{i}} <\select{\smt{b}}{\smt{j}} \quad\text{ or }\quad \mI_2 = \forall \smt{i}. \exists \smt{j}.\ \select{\smt{a}}{\smt{i}} < \select{\smt{b}}{\smt{j}}\enspace.\]
  Both do not lie in the array property fragment due to quantifier alternation.

  In fact, there does not exist an interpolant within the array property fragment.
  Intuitively, the only shared terms that can be used in the interpolant, are the arrays \(\smt{a}\) and \(\smt{b}\).
  There is no shared term to capture the indices \(\smt{k}\) and \(\smt{l}\).
  One can only obtain index terms by using a quantifier, which will lead to quantifier alternation.
\end{example}

It is well known that the quantifier-free fragment of the theory of arrays is not closed under interpolation, but if one adds an auxiliary function the extension allows for interpolation~\cite{DBLP:journals/corr/abs-1204-2386}. 
This can be achieved by adding the $\diff$ function that returns some index where two arrays differ.
The meaning of the $\diff$ function is not fixed; it only needs to satisfy the property
\begin{equation}\label{eq:diff}\tag{$*$}
  s(\diff^{\Model}(s,t)) \neq t(\diff^{\Model}(s,t)) \text{ for all $s,t$ with $s(j) \neq t(j)$ for some $j$}\,.
\end{equation}

However, in the above example the $\diff$ function is not sufficient to define an interpolant without quantifier alternation.
The informal reason is that it is only required to return some index where two arrays differ (if they differ), hence, to capture the correct index, more information on the index, expressible in shared terms, would be needed.
The following theorem will prove this formally.

\begin{theorem}\label{thm:ap_notclosed_thm}
  The array property fragment is not closed under interpolation.
\end{theorem}

\begin{proof}
  Consider again Example~\ref{thm:ap_notclosed_ex}.
  In the following, we show that there does not exist an interpolant without quantifier alternation for this interpolation problem.
  The proof follows the idea of Drews~and~Albarghouti for showing a similar result for EPR~\cite{DBLP:conf/cav/DrewsA16}.
  
  We construct a sequence of models that are alternatingly models for \A and for \B and show that no formula in the array property fragment that contains only shared terms can distinguish between models for \A and models for \B from a certain point on.
  
  For \(i \in \N\), let
  \begin{align*}
	\smt{k}^{\Model_i} &= \smt{l}^{\Model_i} = i\\
	\smt{a}^{\Model_i} (j) &=
	\begin{cases}
	  0							&\text{ if } j \leq 0 \\
	  \lceil \frac{j}{2} \rceil &\text{ if } 0 < j \leq i \\
	  \lceil \frac{i}{2} \rceil &\text{ if } i < j
	\end{cases}\\
	\smt{b}^{\Model_i}(j) &=
	\begin{cases}
	  1									&\text{ if } j \leq 0 \\
	  \lfloor \frac{j}{2} \rfloor + 1	&\text{ if } 0 < j \leq i \\
	  \lfloor \frac{i}{2} \rfloor + 1	&\text{ if } i < j
	\end{cases}\\
	\diff^{\Model_i}(s,t) &=
	\begin{cases}
	  0 &\text{ if }s(j)=t(j)\text{ for all } j\\
	  \mathrm{max}\{j \mid j < 0 \land s(j) \neq t(j)\} &\text{ if }s(j) \neq t(j)\text{ for some }j<0\\
	  \mathrm{min}\{j \mid s(j) \neq t(j)\} &\text{ otherwise}
	\end{cases}
  \end{align*}
  
  For any even number \(i\) (including 0), \(\Model_i\) is a model for \A, and for any odd number \(i\), \(\Model_i\) is a model for \B:
  the maximum value of both \(\smt{a}\) and \(\smt{b}\) is stored at index \(i\).
  If \(i\) is even, \(\lceil \frac{i}{2} \rceil = \lfloor \frac{i}{2} \rfloor\) and hence for \(\smt{k}^{\Model_i} = i\), the value \(\smt{b}^{\Model_i}(i) = \lfloor \frac{i}{2} \rfloor + 1\) is greater than all values in \(\smt{a}^{\Model_i}\).
  If \(i\) is odd, \(\lceil \frac{i}{2} \rceil = \lfloor \frac{i}{2} \rfloor + 1\) and for \(\smt{l}^{\Model_i} = i\), the value \(\smt{a}^{\Model_i}(i) = \lceil \frac{i}{2} \rceil\) is greater or equal to all values in \(\smt{b}^{\Model_i}\).
  
  Note that {\rm max} and {\rm min} in the semantics of $\diff$ are well-defined:
  in the second case, \(\{j \mid j < 0 \land \select{s}{j} \neq \select{t}{j}\}\) is a non-empty set of negative integers, and hence has a maximum element.
  In the last case, \(\{j \mid \select{s}{j} \neq \select{t}{j}\}\) is a non-empty set of non-negative integers and has a minimal element.
  By definition, if \(s \neq t\), \(\diff(s,t)\) returns an index where \(s\) and \(t\) differ.
  Hence, property \eqref{eq:diff} is satisfied.
  
  \medskip
  
  We will now show that any formula in the array property fragment only containing shared symbols cannot distinguish between \(\Model_i\) and \(\Model_{i'}\) for large \(i\) and \(i'\).
  Therefore it cannot be an interpolant of \((A,B)\):
  an interpolant evaluates to true for all even \(i\) and to false for all odd \(i\).
  
  We first consider quantifier-free terms and distinguish between array-valued terms \(t_a\) and scalar terms \(t_s\).
  The latter includes also Boolean terms.
  The following properties hold:
  \begin{enumerate}
	\item For all shared scalar terms \(t_s\), there exists a number \(i\) such that for all models \(\Model_{i'}\) with \(i \leq i'\), the value of \(t_s\) does not change, i.e., \(t_s^{\Model_{i}} = t_s^{\Model_{i'}}\).
	
	\item For all shared array terms \(t_a\), there exists a number \(i\) such that
	\begin{enumerate}
	  \item the prefix of the array \(t_a\) does not change for subsequent models, i.e., for all \(i',i''\) with \(i \leq i' \leq i''\), and for all indices \(j\) with \(j \leq i'\), \(t_a^{\Model_{i'}}(j) = t_a^{\Model_{i''}}(j)\), and
	  
	  \item for all \(i'\) with \(i \leq i'\), the suffix of the array \(t_a\) repeats the element at index \(i'\), i.e., for all \(j\) with \(i' < j\), it holds \(t_a^{\Model_{i'}}(i') = t_a^{\Model_{i'}}(j)\).
	\end{enumerate}
  \end{enumerate}
  Note that if a property holds for one number \(i\), then it also holds for all larger numbers by definition.
  
  We show Properties~1 and 2 by induction over the term \(t_s\) and \(t_a\), respectively.
  
  \emph{Base case:}
  For integer constants, Property~1 holds for all \(i\).
  For the shared terms \(\smt{a}\) and \(\smt{b}\), Property~2 holds for all \(i\).
  
  \emph{Induction step:}
  For function applications and predicates that do not involve arrays, e.g. \(t_s = t_{1} + t_{2}\), we assume that Property~1 holds for \(t_1\) and \(t_2\) with \(i_1\) and \(i_2\).
  Then for \(i := \mathrm{max}\{i_1, i_2\}\), Property~1 holds for \(t_s\).
  
  For a select term \(\select{t_a}{t_{s}}\), we assume that Property~1 holds for \(t_s\) with \(i_1\) and Property~2 holds for \(t_a\) with \(i_2\).
  Then Property~1 holds for \(\select{t_a}{t_{s}}\) with \(i := \mathrm{max}\{i_1, i_2, t_s^{\Model_{i_1}}\}\):
  for all \(i'\) with \(i \leq i'\), we derive \((\select{t_a}{t_s})^{\Model_i} = (\select{t_a}{t_s})^{\Model_{i'}}\) from Property~2(a) since \(i_2 \leq i \leq i'\) and for \(j := t_s^{\Model_i} = t_s^{\Model_{i'}}\), we have \(j \leq i\).
  
  For a term \(\diff(t_1,t_2)\), we assume that Property~2 holds for \(t_1\) and \(t_2\) with \(i_1\) and \(i_2\), respectively, and thus, it holds for both \(t_1\) and \(t_2\) with \(i_0 := \mathrm{max}\{i_1,i_2\}\).
  If for some \(i\) with \(i \geq i_0\), \(t_1^{\Model_i} \neq t_2^{\Model_i}\), then because of Property~2(b), \(t_1\) and \(t_2\) differ at some index \(j\) with \(j \leq i\).
  By definition, \((\diff(t_1,t_2))^{\Model_i} \leq i\).
  Because of Property~2(a) and the definition of $\diff$, for \(i'\) with \(i' > i\), we have \((\diff(t_1,t_2))^{\Model_{i'}} = (\diff(t_1,t_2))^{\Model_i}\).
  If for all \(i\) with \(i \geq i_0\), \(t_1^{\Model_i} = t_2^{\Model_i}\), then \((\diff(t_1,t_2))^{\Model_i} = 0\) by definition, and Property~1 holds for \(\diff(t_1,t_2)\) with \(i_0\).
  
  For a store term \(\store{t_a}{t_0}{t_1}\), we assume that Property~1 holds for \(t_0\) and \(t_1\) with \(i_0\) and \(i_1\), and Property~2 holds for \(t_a\) with \(i_2\).
  With \(i := \mathrm{max}\{i_0, i_1, i_2, {t_0^{\Model_{i_0}} + 1}\}\), Property~2 holds for \(\store{t_a}{t_0}{t_1}\):
  (a) holds for \(j \neq t_0^{\Model_{i_0}}\) because it holds for \(t_a\), and for \(j = t_0^{\Model_{i_0}}\), (a) follows from Property~1 for \(t_1\).
  Property~2(b) holds for \(j > i\) because it holds for \(t_a\) and \(i > t_0^{\Model_{i_0}}\).
  
  \medskip
  
  Next we show that for an array property \(\varphi : \forall \bar{j}.\, \varphi_I(\bar{j}) \bimplies \varphi_V(\bar{j})\), there exists a number \(i\) such that the value of \(\varphi\) stays constant, i.e., \(\varphi^{\Model_{i'}} = \varphi^{\Model_{i}}\) for \(i' \geq i\).
  
  First, we collect all subterms of \(\varphi\) that do not contain \(j\) and compute the corresponding \(i\) that satisfy Property~1 or 2, respectively.
  Let \(i_0\) be the maximum of all these numbers.
  For all ground terms \(t\) in the index guard \(\varphi_I(\bar{j})\), compute \(t^{\Model_{i_0}}\) and let \(i_1\) be the maximum of \(i_0\) and all numbers \(t^{\Model_{i_0}} + 1\).
  
  If for all \(i_2 \geq i_1\), \(\varphi^{\Model_{i_2}}\) is true, the value of \(\varphi\) obviously stays constant in all subsequent models.
  
  If there exists \(i_2 \geq i_1\) such that \(\varphi^{\Model_{i_2}}\) is false, there is some \(\bar{j}\) such that \(\varphi_I(\bar{j}) \bimplies \varphi_V(\bar{j})\) is false under \(\Model_{i_2}\).
  If we replace all components of \(\bar{j}\) that are greater than \(i_2\) by \(i_2\), the formula is still false.
  The index guard is still true: 
  Let \(j_m \in \bar{j}\) be greater than \(i_2\).
  As \(i_2\) is greater than the maximum of all values \(t^{\Model_{i_0}} = t^{\Model_{i_2}}\) for the ground terms \(t\) in the index guard, literals of the form \(j_m \leq t\) or \(j_m = t\) must evaluate to false in \(\Model_{i_2}\).
  For a literal \(t \leq j_m\), the replacement \(t \leq i_2\) will evaluate to true because of the definition of \(i_2\).
  If a literal \(j_m = j_n\) evaluates to true in \(\Model_{i_2}\), then we replace both \(j_m\) and \(j_n\) by \(i_2\) and the resulting equality holds trivially.
  This means, by replacing \(j_m\) by \(i_2\) we can only obtain more literals that evaluate to true in the index guard.
  The evaluation under \(\Model_{i_2}\) of the value guard is unchanged because of Property~2(b) for the arrays in the select terms containing \(j\), and Property~1 for the other terms.
  Note that quantified variables \(j\) cannot appear in store or $\diff$ terms, because array reads \(\select{t_a}{j}\) must not be nested.

  Thus, we can assume that all components of \(\bar{j}\) are smaller or equal to \(i_2\).
  Then, for all \(i'\) with \(i' \geq i_2\), \((\varphi_I(\bar{j}) \bimplies \varphi_V(\bar{j}))^{\Model_{i'}}\) is still false.
  This follows from Property~2(a) for the arrays in select terms, and Property~1 for all other terms.
  Thus, for all \(i'\) with \(i' \geq i_2\), \(\varphi^{\Model_{i'}}\) is constantly false.
  
  \medskip
  
  Every formula in the array property fragment over shared symbols is a Boolean combination of array properties and quantifier-free formulas.
  For each of these formulas, there exists a number \(i\) from which on the formulas do not change their value.
  If we choose the maximum of all these numbers \(i\), the whole formula does not change its value between \(\Model_i\) and \(\Model_{i+1}\) and as one of  \(\Model_i\) and \(\Model_{i+1}\) is a model for \A and the other is a model for \B, the formula cannot be an interpolant for \((A,B)\).
  \qed
\end{proof}

\section{Conclusion}

The array property fragment is an expressive but still decidable fragment for the theory of arrays and therefore useful for checking program correctness.
In this paper, we have shown that the array property fragment is not closed under interpolation.
Our proof also shows that, in contrast to the quantifier-free fragment, the $\diff$ function does not establish closedness under interpolation for the array property fragment.
Thus, it is not sufficient to restrict the solver to the array property fragment, if one wants to use interpolants to derive new invariants used in later solver queries.

As our example shows, the problem is that the array property fragment cannot express interpolants of simple quantified formulas.
Therefore, for interpolation based software model checking a more expressive fragment is needed.
One possible candidate is the almost uninterpreted fragment~\cite{DBLP:conf/cav/GeM09}.
This fragment allows for quantifier alternation and can express the interpolants in our example.
However, this fragment is undecidable.  
One can achieve decidability by using the finite almost uninterpreted fragment, however, this fragment also does not have nice closure properties: it is not even closed under conjunction.

\bibliographystyle{plain}
\bibliography{references}

\end{document}